\pdfoutput=1
\documentclass[11pt]{article}
\usepackage{hyperref}       
\usepackage{url}            
\usepackage{amsmath,amssymb,amsthm,amsfonts}
\usepackage{hyperref}
\usepackage{enumitem}
\usepackage{multirow}

\usepackage{algorithm}
\usepackage{algorithmic}
\usepackage[margin=1in]{geometry}

\theoremstyle{plain}
\newtheorem{theorem}{Theorem}[section]
\newtheorem{lemma}[theorem]{Lemma}

\theoremstyle{definition}
\newtheorem{definition}{Definition}[section]
\newtheorem{remark}{Remark}
\newcommand{\ignore}[1]{}

\newcommand{\eps}{\varepsilon}

\newcommand{\R}{\ensuremath\mathbb{R}}

\newcommand{\norm}[1]{\left\| #1\right\|}

\DeclareMathOperator{\rk}{rank}

\DeclareMathOperator{\diag}{diag}
\DeclareMathOperator{\nnz}{nnz}
\DeclareMathOperator{\poly}{poly}
\DeclareMathOperator{\argmin}{arg\,min}

\title{Input-Sparsity Low Rank Approximation in Schatten Norm\footnote{Y.\@ Li was supported in part by Singapore Ministry of Education (AcRF) Tier 2 grant MOE2018-T2-1-013. D.\@ P.\@ Woodruff was supported in part by Office of Naval Research (ONR) grant N00014-18-1-2562. Part of this work was done during a visit of D.\@ P.\@ Woodruff to Nanyang Technological University, funded by the aforementioned Singapore Ministry of Education grant.}}
\author{Yi Li\\ Nanyang Technological University\\ \texttt{yili@ntu.edu.sg} \and
	   David P. Woodruff\\ Carnegie Mellon University\\ \texttt{dwoodruf@andrew.cmu.edu}}

\date{}

\begin{document}

\maketitle

\allowdisplaybreaks

\begin{abstract}
We give the first input-sparsity time algorithms for the rank-$k$ low rank approximation problem in every Schatten norm. Specifically, for a given $m\times n$ ($m\geq n$) matrix $A$, our algorithm computes $Y\in \R^{m\times k}$, $Z\in \R^{n\times k}$, which, with high probability, satisfy $\|A-YZ^T\|_p \leq (1+\eps)\|A-A_k\|_p$, where $\|M\|_p = \left (\sum_{i=1}^n \sigma_i(M)^p \right )^{1/p}$ is the Schatten $p$-norm of a matrix $M$ with singular values $\sigma_1(M), \ldots, \sigma_n(M)$, and where $A_k$ is the best rank-$k$ approximation to $A$. Our algorithm runs in time $\tilde{O}(\nnz(A) + mn^{\alpha_p}\poly(k/\eps))$, where $\alpha_p = 0$ for $p\in [1,2)$ and $\alpha_p = (\omega-1)(1-2/p)$ for $p>2$ and $\omega \approx 2.374$ is the exponent of matrix multiplication. For the important case of $p = 1$, which corresponds to the more ``robust'' nuclear norm, we obtain $\tilde{O}(\nnz(A) + m \cdot \poly(k/\epsilon))$ time, which was previously only known for the Frobenius norm $(p = 2)$. Moreover, since $\alpha_p < \omega - 1$ for every $p$, our algorithm has a better dependence on $n$ than that in the singular value decomposition for every $p$. Crucial to our analysis is the use of dimensionality reduction for Ky-Fan $p$-norms. 
\end{abstract}

\section{Introduction}
A common task in processing or analyzing large-scale datasets is to approximate a large matrix $A\in \R^{m\times n}$ ($m\geq n$) with a low-rank matrix. Often this is done with respect to the Frobenius norm, that is, the objective function is to minimize the error $\norm{A-X}_F$ over all rank-$k$ matrices $X\in \R^{m\times n}$ for a rank parameter $k$. It is well-known that the optimal solution is $A_k = P_LA = AP_R$, where $P_L$ is the orthogonal projection onto the top $k$ left singular vectors of $A$, and $P_R$ is the orthogonal projection onto the top $k$ right singular vectors of $A$. Typically this is found via the singular value decomposition (SVD) of $A$, which is an expensive operation. 

For large matrices $A$ this is too slow, so we instead allow for randomized approximation algorithms in the hope of achieving a much faster running time. Formally, given an approximation parameter $\eps > 0$, we would like to find a rank-$k$ matrix $X$ for which $\norm{A-X}_F\leq (1+\eps)\norm{A-A_k}_F$ with large probability. For this relaxed problem, a number of efficient methods are known, which are based on dimensionality reduction techniques such as random projections, importance sampling, and other sketching methods, with running times\footnote{We use the notation $\tilde{O}(f)$ to hide the polylogarithmic factors in $O(f\poly(\log f))$.}$^{,}$\footnote{Since outputting $X$ takes $O(mn)$ time, these algorithms usually output $X$ in factored form, where each factor has rank $k$.} $\tilde{O}(\nnz(A) + m\poly(k/\eps))$, where $\nnz(A)$ denotes the number of non-zero entries of $A$. This is significantly faster than the SVD, which takes $\tilde{\Theta}(mn^{\omega-1})$ time, where $\omega$ is the exponent of matrix multiplication. See \cite{woodruff:book} for a survey.

In this work, we consider approximation error with respect to general matrix norms, i.e., to the Schatten $p$-norm. The Schatten $p$-norm, denoted by $\|\cdot\|_p$, is defined to be the $\ell_p$-norm of the singular values of the matrix. Below is the formal definition of the problem.

\begin{definition}[Low-rank Approximation]
Let $p\geq 1$. Given a matrix $A\in \R^{m\times n}$, find a rank-$k$ matrix $\hat X\in\R^{m\times n}$ for which 
\begin{equation}\label{eqn:low_rank_approx}
\norm{A-\hat X}_p \leq (1+\eps)\min_{X:\rk(X)=k} \norm{A-X}_p.
\end{equation}
\end{definition}
It is a well-known fact (Mirsky's Theorem) that the optimal solution for general Schatten norms coincides with the optimal rank-$k$ matrix $A_k$ for the Frobenius norm, given by the SVD. However, approximate solutions for the Frobenius norm loss function may give horrible approximations for other Schatten $p$-norms. 

Of particular importance is the Schatten $1$-norm, also called the nuclear norm or the trace norm, which is the sum of the singular values of a matrix. It is typically considered to be more robust than the Frobenius norm (Schatten $2$-norm) and has been used in robust PCA applications (see, e.g., \cite{xu:robustPCA,candes:robustPCA,yi:robustPCA}). 

For example, suppose the top singular value of an $n \times n$ matrix $A$ is $1$, the next $2k$ singular values are $1/\sqrt{k}$, and the remaining singular values are $0$. A Frobenius norm rank-$k$ approximation could just choose the top singular direction and pay a cost of $\sqrt{2k \cdot 1/k} = \sqrt{2}$. Since the Frobenius norm of the bottom $n-k$ singular values is $(k+1) \cdot \frac{1}{k}$, this is a $\sqrt{2}$-approximation. On the other hand, if a Schatten $1$-norm rank-$k$ approximation algorithm were to only output the top singular direction, it would pay a cost of $2k \cdot 1/\sqrt{k} = 2\sqrt{k}$. The bottom $n-k$ singular values have Schatten $1$-norm $(k+1) \cdot \frac{1}{\sqrt{k}}$. Consequently, the approximation factor would be $2(1-o(1))$, and one can show if we insisted on a $\sqrt{2}$-approximation or better, a Schatten $1$-norm algorithm would need to capture a constant fraction of the top $k$ directions, and thus capture more of the underlying data than a Frobenius norm solution. 

Consider another example where the top $k$ singular values are all $1$s and the $(k+i)$-th singular value is $1/i$. When $k=o(\log n)$, capturing only the top singular direction will give a $(1+o(1))$-approximation for the Schatten $1$-norm but this gives a $\Theta(\sqrt{k})$-approximation for the Frobenius norm. This example, together with the previous one, shows that the Schatten norm is a genuinely different error metric.

Surprisingly, no algorithms for low-rank approximation in the Schatten $p$-norm were known to run in time $\tilde{O}(\nnz(A) + m\poly(k/\eps))$ prior to this work, except for the special case of $p=2$. We note that the case of $p = 2$ has special geometric structure that is not shared by other Schatten $p$-norms. Indeed, a common technique for the $p = 2$ setting is to first find a $\poly(k/\epsilon)$-dimensional subspace $V$ containing a rank-$k$ subspace inside of it which is a $(1+\epsilon)$-approximate subspace to project the rows of $A$ on. Then, by the Pythagorean theorem, one can first project the rows of $A$ onto $V$, and then find the best rank-$k$ subspace of the projected points inside of $V$. For other Schatten $p$-norms, the Pythagorean theorem does not hold, and it is not hard to construct counterexamples to this procedure for $p \neq 2$. 

To summarize, the SVD runs in time $\Theta(mn^{\omega-1})$, which is much slower than $\nnz(A)\leq mn$. It is also not clear how to adapt existing fast Frobenius-norm algorithms to generate $(1+\eps)$-factor approximations with respect to other Schatten $p$-norms. 

\paragraph{Our Contributions} In this paper we obtain the first provably efficient algorithms for the rank-$k$ $(1+\eps)$-approximation problem with respect to the Schatten $p$-norm for every $p\geq 1$. We describe our results for square matrices below. Our general results for rectangular matrices can be found in the precise statements of the theorems.
\begin{theorem}[informal, combination of Theorems~\ref{thm:p<2} and \ref{thm:p>2}]
Suppose that $m\geq n$ and $A\in\R^{m\times n}$. There is a randomized algorithm which outputs two matrices $Y\in \R^{m\times k}$ and $Z\in \R^{n\times k}$ for which $\hat X = Y Z^T$ satisfies \eqref{eqn:low_rank_approx} with probability at least $0.9$. The algorithm runs in time $O(\nnz(A)\log n) +\tilde{O}(mn^{\alpha_p} \poly(k/\eps))$, where
\[
\alpha_p = \begin{cases}
			0, & 1\leq p\leq 2;\\
			(\omega-1)(1-\frac2p), & p > 2,
		  \end{cases}	
\]
and the hidden constants depend only on $p$.
\end{theorem}

In the particular case of $p=1$, and more generally for all $p\in[1,2]$, our algorithm achieves a running time of $O(\nnz(A)\log n + m\poly(k/\eps))$, which was previously known to be possible for $p=2$ only. When $p>2$, the running time begins to depend polynomially on $n$ but the dependence remains $o(n^{\omega-1})$ for all larger $p$. Thus, even for larger values of $p$, when $k$ is subpolynomial in $n$, our algorithm runs substantially faster than the SVD. Empirical evaluations are also conducted to demonstrate our improved algorithm when $p=1$ in Section~\ref{sec:experiments}.

It was shown by Musco and Woodruff~\cite{MW17} that computing a constant-factor low-rank approximation to $A^TA$, given only $A$, requires $\Omega(\nnz(A) \cdot k)$ time. Given that the squared singular values of $A$ are the singular values of $A^TA$, it is natural to suspect that obtaining a constant-factor low rank approximation to the Schatten $4$-norm low-rank approximation would therefore require $\Omega(\nnz(A) \cdot k)$ time. Surprisingly, we show this is not the case, and obtain an $\tilde{O}(\nnz(A) + mn^{(\omega-1)/2}\poly(k/\eps))$ time algorithm. This does not contradict the aforementioned lower bound as it is not clear how to produce efficiently a low-rank approximation to $A^TA$ in the Frobenius norm from a low-rank approximation to $A$ in the Schatten $4$-norm\footnote{Our result finds an orthogonal projection $Q$ onto a $k$-dimensional subspace such that $\norm{A-AQ}_4 \approx \norm{A-A_k}_4$, so $(I-Q)^TA^TA(I-Q)$ is a good residual for the approximation to $A^TA$ in the Frobenius norm. But $A^TA - (I-Q)^TA^TA(I-Q)$ may not have rank at most $k$, and the easily computable candidate $Q^TA^TAQ$ may not be a good approximation to $A^TA$. Consider $A = \left(\begin{smallmatrix} 20 & 20\\ 1 & 2 \end{smallmatrix}\right)$, $v=(\frac{1}{\sqrt{2}}, \frac{1}{\sqrt{2}})^T$ and $Q = vv^T$ (projection onto the subspace spanned by $\{v\}$). Then $A$ has two singular values $28.3637$ and $0.7051$, $A(I-Q)$ has one singular value $0.7071$, which implies that it is a good rank-$1$ approximation to $A$, while $A^T A - Q^T A^T A Q$ has two singular values $1.7707$ and $1.2707$, hence $Q^TA^TAQ$ cannot be a good rank-$1$ approximation to $A^TA$. The other quantity $A^TA - (I-Q)^TA^TA(I-Q)$ has rank $2$ with eigenvalues $804.5$ and $-0.003$.}.

In addition, we generalize the error metric from matrix norms to a wide family of general loss functions, see Section~\ref{sec:generalization} for details. Thus, we considerably broaden the class of loss functions for which input sparsity time algorithms were previously known for. 

\paragraph{Technical Overview.} We illustrate our ideas for $p=1$. Our goal is to find an orthogonal projection $\hat Q'$ for which $\big\|A(I-\hat Q')\big\|_1 \leq (1+O(\eps))\norm{A-A_k}_1$. The crucial idea in the analysis is to split $\|\cdot\|_1$ into a head part $\|\cdot\|_{(r)}$, which, known as the Ky-Fan norm, equals the sum of the top $r$ singular values, and a tail part $\|\cdot\|_{(-r)}$ (this is just a notation---the tail part is not a norm), which equals the sum of all the remaining singular values. Observe that for $r\geq k/\eps$ it holds that $\big\|A(I-\hat Q')\big\|_{(-r)} \leq \norm{A}_{(-r)} \leq \norm{A-A_k}_{(-r)} + \eps\norm{A-A_k}_1$ for any rank-$k$ orthogonal projection $\hat Q'$ and it thus suffices to find $\hat Q'$ for which $\big\|A(I-\hat Q')\big\|_{(r)} \leq (1+\eps)\norm{A-A_k}_{(r)}$. 
To do this, we sketch $A(I-Q)$ on the left by a projection-cost preserving matrix $S$ by Cohen et al.~\cite{CMM17} such that $\norm{SA(I-Q)}_{(r)} = (1\pm\eps)\norm{A(I-Q)}_{(r)} \pm \eps\|A-A_k\|_1$ for all rank-$k$ projections $Q$. 
Then we solve $\min_Q \|SA(I-Q)\|_{(r)}$ over all rank-$k$ projections $Q$ and obtain a $(1+\eps)$-approximate projection $\hat Q'$, which, intuitively, is close to the best projection $P_R$ for $\min_Q \norm{A(I-Q)}_{(r)}$, and can be shown to satisfy the desired property above.

The last step is to approximate $A\hat Q'$, which could be expensive if done trivially, so we reformulate it as a regression problem $\min_Y \norm{A-YZ^T}_1$ over $Y\in \R^{n\times k}$, where $Z$ is an $n\times k$ matrix whose columns form an orthonormal basis of the target space of the projection $\hat Q'$. This latter idea has been applied successfully for Frobenius-norm low-rank approximation (see, e.g., \cite{CW17}). Here we need to argue that the solution to the Frobenius-norm regression $\min_Y \norm{A-YZ^T}_F$ problem gives a good solution to the Schatten $1$-norm regression problem. Finally we output $Y$ and $Z$.

\section{Preliminaries}

\paragraph{Notation} For an $m\times n$ matrix $A$, let $\sigma_1(A)\geq \sigma_2(A)\geq \cdots \geq \sigma_{s}(A)$ denote its singular values, where $s = \min\{m,n\}$. The Schatten $p$-norm ($p\geq 1$) of $A$ is defined to be $\norm{A}_p := \sum_{i=1}^s (\sigma_i(A)^p)^{1/p}$ and the singular $(p,r)$-norm ($r\leq s$) to be $\norm{A}_{(p,r)} = \sum_{i=1}^r (\sigma_i(A)^p)^{1/p}$. It is clear that $\|A\|_p = \|A\|_{(p,s)}$. When $p=2$, the Schatten $p$-norm coincides with the Frobenius norm and we shall use the notation $\norm{\cdot}_F$ in preference to $\norm{\cdot}_2$.

Suppose that $A$ has the singular value decomposition $A = U\Sigma V^T$, where $\Sigma$ is a diagonal matrix of the singular values. For $k \leq \min\{m,n\}$, let $\Sigma_k$ denote the diagonal matrix for the largest $k$ singular values only, i.e., $\Sigma_k = \diag(\sigma_1(A),\dots,\sigma_k(A),0,\dots,0)$. We define $A_k = U\Sigma_k V^T$. The famous Mirsky's theorem states that $A_k$ is the best rank-$k$ approximation to $A$ for any rotationally invariant matrix norm.

For a subspace $E\subseteq \R^n$, we define $P_E$ to be an $n\times \dim(E)$ matrix whose columns form an orthonormal basis of $E$.

\paragraph{Toolkit} There has been extensive research on randomized numerical linear algebra in recent years. Below are several existing results upon which our algorithm will be built.

\begin{definition}[Sparse Embedding Matrix]
Let $\eps > 0$ be an error parameter. The $(n,\eps)$-sparse embedding matrix $R$ of dimension $n\times r$ is constructed as follows, where $r$ is to be specified later. Let $h:[n]\to [r]$ be a random function and $\sigma:[n]\to \{-1,1\}$ be a random function. The matrix $R$ has only $n$ nonzero entries: $R_{i,h(i)} = \sigma(i)$ for all $i\in [n]$. 
The value of $r$ is chosen to be $r=\Theta(1/\eps^2)$ such that 
\[
\Pr_R\{\norm{A^TRR^TB-A^TB}_F^2\leq \eps^2\norm{A}_F^2\norm{B}_F^2\}\geq 0.99
\] for all $A$ with orthonormal columns. 
This is indeed possible by~\cite{CW17,mm13,NN13}. 
\end{definition}
It is clear that, for a matrix $A$ with $n$ columns and an $(n,\eps)$-sparse embedding matrix $R$, the matrix product $AR$ can be computed in $O(\nnz(A))$ time.

\begin{lemma}[Thin SVD~\cite{demmel}]
\label{lem:thin_svd}
Suppose that $A\in \R^{m\times n}$ with $m\geq n$ and the (thin) singular value decomposition is $A=U\Sigma V^T$, where $U\in\R^{m\times n}$ and $\Sigma, V\in \R^{n\times n}$. Computing the full thin SVD takes time $\tilde{O}(mn^{\omega-1})$.
\end{lemma}

\begin{lemma}[Multiplicative Spectral Approximation~\cite{CLMMPS15}] \label{lem:spec_approx}
Suppose that $A\in \R^{m\times n}$ ($n\leq m \leq \poly(n)$) has rank $r$. There exists a sampling matrix $R$ of $O(\eps^{-2}r\log r)$ rows such that $(1-\eps)A^TA\preceq (RA)^T(RA)\preceq (1+\eps)A^TA$ with probability at least $0.9$ and $R$ can be computed in $O(\nnz(A)\log n + n^\omega \log^2 n + n^{2}\eps^{-2})$ time, where $\theta$ is an arbitrary constant in $(0,1]$.
\end{lemma}

\begin{lemma}[Additive-Multiplicative Spectral Approximation~\cite{CMM17,christopher:thesis}]\label{lem:add-mul_spec_approx}
Suppose that $A\in \R^{m\times n}$ with $m\geq n$, error parameters $\eps\geq \eta \geq 1/\poly(n)$. Let $K = k + \eps/\eta$.
There exists a randomized algorithm which runs in $O(\nnz(A)\log n) + \tilde{O}(mK^{\omega-1})$ time and outputs a matrix $C$ of $t = \Theta(\eps^{-2}K\log K)$ columns, which are rescaled column samples of $A$ without replacement, such that with probability at least $0.99$,
\[
 (1-\eps)AA^T - \eta\norm{A-A_k}_F^2 I \preceq CC^T \preceq (1+\eps)AA^T + \eta\norm{A-A_k}_F^2 I.
\]
\end{lemma}

We also need an elementary inequality.
\begin{lemma}\label{lem:elementary}
Suppose that $p\geq 1$ and $\eps \in (0,1]$. Let $C_{p,\eps} = p(1+1/\eps)^{p-1}$. It holds for $x\in[\eps,1]$ that $(1 + x)^p\leq 1 + C_{p,\eps} x^p$ and that $(1-x)^p \geq 1 - C_{p,\eps} x^p$.
\end{lemma}
\begin{proof}
It is easy to see that for $x\in [\eps,1]$,
\[
(1+x)^p \leq 1 + \frac{(1+\eps)^p-1}{\eps^p} x^p
\quad \text{and}\quad
(1-x)^p \geq 1 - \frac{1 - (1+\eps)^p}{\eps^p} x^p.
\]
Then note that 
\[
p\left(1 + \frac{1}{\eps}\right)^{p-1} \geq \frac{(1+\eps)^p-1}{\eps^p} \geq \frac{1 - (1+\eps)^p}{\eps^p}. \qedhere
\]
\end{proof}

\begin{algorithm}[!t]
\caption{Outline of the algorithm for finding a low-rank approximation}\label{alg:outline}
\begin{algorithmic}[1]
	\IF{$p<2$}
		\STATE $\eta_1 \gets O((\eps^2/k)^{2/p})$, $\eta_2 \gets O(\eps^2/k^{2/p-1})$
	\ELSE
		\STATE $\eta_1 \!\gets\! O(\eps^{1+2/p}/k^{2/p}n^{1-2/p})$, $\eta_2 \!\gets\! O(\eps^2/n^{1-2/p})$
	\ENDIF
	\STATE Use Lemma~\ref{lem:add-mul_spec_approx} to obtain a sampling matrix $S$ of $s$ rows such that 
	\begin{equation}\label{eqn:PCP_eta1}
	 (1-\eps)A^TA - \eta_1\norm{A-A_k}_F^2 I 
	\preceq A^T S^T SA 
	\preceq (1+\eps)A^TA + \eta_1\norm{A-A_k}_F^2 I.
	\end{equation}

	\STATE $T\gets $ subspace embedding matrix for $s$-dimensional subspaces with error $O(\eps)$
	\STATE $W' \gets $ projection onto the top $k$ left singular vectors of $SAT$ \label{alg:line W'}
	\STATE $Z \gets $ matrix whose columns are an orthonormal basis of the row space of $W'SA$ \label{alg:line Z}
	\STATE $R \gets (n,\Theta(\sqrt{\eta_2/k}))\text{-sparse embedding matrix}$
	\STATE $\hat Y \gets ARP_{\text{rowspace}(Z^TR)}$, where the projector $P$ has $k$ columns \label{alg:line Y hat}
	\RETURN $\hat Y, Z$
\end{algorithmic}
\end{algorithm}

\section{Case $p<2$}
The algorithm is presented in Algorithm~\ref{alg:outline}. In this section we shall prove the correctness and analyze the running time for a constant $p\in [1,2)$. Throughout this section we set $r = k/\eps$.

\begin{lemma}\label{lem:(p,r)-norm-preserved}
Suppose that $p\in (0,2)$ and $S$ satisfies \eqref{eqn:PCP_eta1}. It then holds for all rank-$k$ orthogonal projections $Q$ that
\begin{align*}
&\quad\ (1-\eps)\norm{A(I-Q)}_{(p,r)}^p - r\eta_1^{\frac{p}{2}}\norm{A-A_k}_p^p \\
&\leq \norm{SA(I-Q)}_{(p,r)}^p \\
&\leq (1+\eps)\norm{A(I-Q)}_{(p,r)}^p + r\eta_1^{\frac{p}{2}}\norm{A-A_k}_p^p.
\end{align*}
\end{lemma}
\begin{proof}
Since $S$ satisfies \eqref{eqn:PCP_eta1}, it holds for any rank-$k$ orthogonal projection $Q$ that
\begin{align*}
&\quad\ (1-\eps)(I-Q)A^TA(I-Q) - \eta_1\norm{A-A_k}_F^2 I \\
&\preceq (I-Q)A^T S^T SA(I-Q) \\
&\preceq (1+\eps)(I-Q)A^TA(I-Q) + \eta_1\norm{A-A_k}_F^2 I.
\end{align*}

The following relationship between singular values of $SA(I-Q)$ and $A(I-Q)$ is an immediate corollary via the max-min characterization of singular values (cf., e.g., Lemma 7.2 of \cite{LNW19})
\begin{equation}\label{eqn:key_inequality}
\begin{aligned}
&\quad\ (1-\eps)\sigma_i^2(A(I-Q)) - \eta_1\norm{A-A_k}_F^2  \\
&\leq \sigma_i^2(SA(I-Q)) \\
&\leq (1+\eps)\sigma_i^2(A(I-Q)) + \eta_1\norm{A-A_k}_F^2
\end{aligned}
\end{equation}

Since $p<2$ and thus $\norm{\cdot}_F\leq \norm{\cdot}_p$, we have from~\eqref{eqn:key_inequality} that
\begin{align*}
&\quad\ (1-\eps)\sigma_i^p(A(I-Q)) - \eta_1^{\frac{p}{2}}\norm{A-A_k}_p^p\\
&\leq \sigma_i^p(SA(I-Q)) \\
&\leq (1+\eps)\sigma_i^p(A(I-Q)) + \eta_1^{\frac{p}{2}}\norm{A-A_k}_p^p.
\end{align*}
Passing to the $(p,r)$-singular norm yields the desired result.
\end{proof}

\begin{lemma}\label{lem:find Q hat}
When $p\in (0,2)$ is a constant and $\eps\in (0,1/2]$, let $\hat Q' = ZZ^T$ be the projection onto the column space of $Z$, where $Z$ is as defined in Line~\ref{alg:line Z} of Algorithm~\ref{alg:outline}. With probability at least $0.99$, it holds that
\begin{equation}\label{eqn:hat Q'}
\norm{SA(I-\hat Q')}_{(p,r)} \leq (1+\eps)\min_{Q}\norm{SA(I-Q)}_{(p,r)},
\end{equation}
where the minimization on the right-hand side is over all rank-$k$ orthogonal projections $Q$.
\end{lemma}
\begin{proof}
Observe that
\[
\min_Q \norm{SA-SAQ}_{(p,r)} = \min_W \norm{SA-WSA}_{(p,r)},
\]
where the minimizations are over all rank-$k$ orthogonal projections $Q$ and all rank-$k$ orthogonal projections $W$, and the equality is achieved when $Q$ is the projection onto the right top $k$ singular vectors of $SA$ and $W$ the left top $k$ singular vectors.

Since $T$ is an oblivious subspace embedding matrix and preserves all singular values of $(I-W)SA$ up to a factor of $(1\pm\eps)$, we have
\[
\min_{W} \norm{SAT\!-\!{W}SAT}_{(p,r)} \!=\! (1\pm \eps)\min_W \norm{SA \!-\! WSA}_{(p,r)}.
\]
The minimization on the left-hand side above is easy to solve: the minimizer $W'$ is exactly the projection onto the top $k$ singular vectors of $SAT$, as computed in Line~\ref{alg:line W'} of Algorithm~\ref{alg:outline}. Since $\hat Q'$ is the projection onto the row space of $W'SA$, it holds that the row space of $SA\hat Q'$ is the closest space to that of $SA$ in the row space of $W'SA$. Hence
\[
\norm{SA-SA\hat Q'}_{(p,r)} \leq \norm{SA-W'SA}_{(p,r)}.
\]
The claimed result \eqref{eqn:hat Q'} then follows from
\begin{align*}
 \norm{SA-W'SA}_{(p,r)} &\leq \frac{1}{1-\eps} \norm{SAT-W'SAT}_{(p,r)} \\
&= \frac{1}{1-\eps}\min_W \norm{SAT-WSAT}_{(p,r)} \\
&\leq \frac{1+\eps}{1-\eps}\min_W \norm{SA-WSA}_{(p,r)} \\
&\leq (1+4\eps)\min_Q \norm{SA-SAQ}_{(p,r)}
\end{align*}
and rescaling $\eps$. 
\end{proof}

\begin{lemma}\label{lem:projection-away for A}
Let $\eps \in (0,1/2]$. Suppose that $\hat Q'$ satisfies \eqref{eqn:hat Q'}. Then it holds that 
\[
\norm{A(I-\hat Q')}_{(p,r)}^p \leq (1+c_1\eps)\norm{A-A_k}_{(p,r)}^p \\ + c_2k\eta_1^{p/2}\norm{A-A_k}_p^p.
\]
for some absolute constants $c_1,c_2>0$.
\end{lemma}
\begin{proof}
Let $\hat Q = \argmin_Q \norm{SA(I-Q)}_{(p,r)}$, where the minimization is over all rank-$k$ projections $Q$. Let $Q^\ast$ be the orthogonal projection onto the top $k$ right singular vectors of $A$. It follows that
\begin{align*}
 \norm{A(I-\hat Q')}_{(p,r)}^p 
&\leq \frac{1}{1-\eps} \norm{SA(I-\hat Q')}_{(p,r)}^p + \frac{1}{1-\eps} k\eta_1^{\frac{p}{2}}\norm{A-A_k}_p^p \\
&\leq\frac{(1+\eps)^p}{1-\eps} \norm{SA(I-\hat Q)}_{(p,r)}^p + \frac{1}{1-\eps} k\eta_1^{\frac{p}{2}}\norm{A-A_k}_p^p \\
&\leq\frac{(1+\eps)^p}{1-\eps} \norm{SA(I-Q^\ast)}_{(p,r)}^p + \frac{1}{1-\eps} k\eta_1^{\frac{p}{2}}\norm{A-A_k}_p^p \\
&\leq \frac{(1+\eps)^p}{1-\eps}\left( (1+\eps)\norm{A(I-Q^\ast)}_{(p,r)}^p + k\eta_1^{\frac{p}{2}}\norm{A-A_k}_p^p \right) 
 + \frac{1}{1-\eps} k\eta_1^{\frac{p}{2}}\norm{A-A_k}_p^p\\
&= \frac{(1\!+\!\eps)^{p+1}}{1-\eps}\norm{A\!-\!A_k}_{(p,r)}^p + \frac{(1\!+\!\eps)^p\!+\!1}{1-\eps}k\eta_1^{\frac{p}{2}}\norm{A\!-\!A_k}_p^p,
\end{align*}
where the first inequality follows from Lemma~\ref{lem:(p,r)-norm-preserved}, the second inequality Lemma~\ref{lem:find Q hat}, the third inequality follows from the optimality of $\hat Q$ and the fourth inequality again from Lemma~\ref{lem:(p,r)-norm-preserved}.
\end{proof}

The next lemma is an immediate corollary of the preceding lemma.
\begin{lemma}\label{lem:hat Q is good}
Let $\eps \in (0,1/2]$. Suppose that $\hat Q'$ satisfies \eqref{eqn:hat Q'}. Then it holds for some absolute constants $c_1,c_2>0$ that
\[
\norm{A(I-\hat Q')}_p^p \leq (1+c_1\eps)\norm{A-A_k}_p^p
\]
whenever $\eta_1 \leq (\eps^2/(c_2k))^{2/p}$.
\end{lemma}
\begin{proof}
Again let $\hat Q = \argmin_Q \norm{SA(I-Q)}_{(p,r)}$, where the minimization is over all rank-$k$ projections $Q$. Observe that
\begin{align*}
\norm{A(I-\hat Q')}_p^p &= \norm{A(I-\hat Q')}_{(p,r)}^p + \sum_{i\geq r+1} \sigma_i^p(A(I-\hat Q')) \\
					 &\leq (1+c_1\eps)\norm{A-A_k}_{(p,r)}^p + c_2r\eta_1^{p/2}\norm{A-A_k}_p^p + \sum_{i\geq r+1} \sigma_i^p(A) \\
					 &\leq (1+c_1\eps)\norm{A-A_k}_p^p + c_2r\eta_1^{p/2}\norm{A-A_k}_p^p + \sum_{i=r+1}^{r+k+1} \sigma_i^p(A)\\
					 &\leq (1+c_1\eps)\norm{A-A_k}_p^p + c_2r\eta_1^{p/2}\norm{A-A_k}_p^p + \frac{k}{r} \|A-A_k\|_p^p \\
					 &\leq (1+(c_1+1)\eps)\norm{A-A_k}_p^p + c_2r\eta_1^{p/2}\norm{A-A_k}_p^p
\end{align*}
where we used the preceding lemma (Lemma~\ref{lem:projection-away for A}) in the first inequality and $r = k/\eps$ in the last inequality. The claimed result holds when $\eta \leq (\eps/(c_3r))^{2/p}$.
\end{proof}

So far we have found a rank-$k$ orthogonal projection $\hat Q' = ZZ^T$ such that
\[
\norm{A-A\hat Q'}_p \leq (1+c_1\eps)\norm{A-A_k}_p
\]
for some absolute constant $c_1$. However, it is not clear how to compute the matrix product $A\hat Q'$ efficiently. Hence we consider the regression problem
\[
\min_{Y: \rk(Y)= k} \norm{A-YZ^T}_p.
\]
It is clear that the minimizer is $Y=AZ$, which satisfies that $\norm{A-YZ^T}_p = \norm{A-A\hat Q'}_p$, since the rowspace of $YZ^T$ is a $k$-dimensional subspace of the rowspace of $Z^T$ and thus it is exactly the rowspace of $Q$. The next lemma shows that $\hat Y$ is an approximation to $Y$.

\begin{lemma}\label{lem: Y hat}
When $1\leq p<2$ is a constant, the matrix $\hat Y$ defined in Line~\ref{alg:line Y hat} of Algorithm~\ref{alg:outline} satisfies with probability at least $0.99$ that
\[
\norm{A-\hat{Y} Z^T}_p \leq (1+c\eps)\min_{Y: \rk(Y)= k} \norm{A-YZ^T}_p,
\]
for some absolute constant $c>0$, whenever $\eta_2 \leq \eps^2/(2k)^{2/p-1}$.
\end{lemma}
\begin{proof}
First, it is clear that the optimal solution to $\min_Y \norm{A - YZ^T}_p$ is $Y = AZ$, where the minimization is over all rank-$k$ $n\times k$ matrices $Y$.

Note that
\[
\hat Y = ARP_{\text{rowspace}(Z^TR)}
\]
is the minimizer to the Frobenius-norm minimization problem $\min_Y \| (A-YZ^T)R \|_F$. Since $R$ is a sparse embedding matrix of error $\Theta(\sqrt{\eta_2/k})$, one can show that (see, e.g., Lemma 7.8 of \cite{CW17}) with probability at least $0.99$,
\[
\norm{AZ - \hat Y}_F \leq \sqrt{\eta_2} \norm{A-AZZ^T}_F.
\]
It follows that
\begin{align*}
    \norm{A-\hat{Y} Z^T}_p &\leq \norm{A-AZZ^T}_p + \norm{\hat{Y}Z^T - AZZ^T}_p \\
					    &\leq (1+c_1\eps)\norm{A-AZZ^T}_p + \norm{\hat{Y} - AZ}_p \\
					    &\leq (1+c_1\eps)\norm{A-AZZ^T}_p + (2k)^{\frac1p-\frac12}\norm{\hat{Y} - AZ}_F \\
					    &\leq (1+c_1\eps)\norm{A\!-\!AZZ^T}_p + (2k)^{\frac1p-\frac12}\sqrt{\eta_2}\norm{A \!-\! AZZ^T}_F\\
					    &= (1+c_1\eps)\norm{A\!-\!AZZ^T}_p + (2k)^{\frac1p-\frac12}\sqrt{\eta_2} \norm{A \!-\! AZZ^T}_p\\
					    &\leq (1+(c_1+1)\eps)\norm{A-AZZ^T}_p
\end{align*}
provided that $\eta_2 \leq \eps^2/(2k)^{2/p-1}$. 
\end{proof}
\begin{remark}
The preceding lemma (Lemma~\ref{lem: Y hat}) may be of independent interest, as it solves Schatten $p$-norm regression efficiently, which has not been discussed in the literature in the context of dimensionality reduction before.
\end{remark}

In summary, we conclude the section with our main theorem.
\begin{theorem}\label{thm:p<2}
Let $p\in [1,2)$. Suppose that $A\in\R^{m\times n}$ with $m\geq n$. There is a randomized algorithm which outputs $Y\in\R^{m\times k}$ and $Z\in \R^{n\times k}$ such that $\hat X = Y Z^T$ satisfies \eqref{eqn:low_rank_approx} with probability at least $0.97$. The algorithm runs in time $O(\nnz(A)\log n) + \tilde{O}_p(mk^{2(\omega-1)/p}/\eps^{(4/p-1)(\omega-1)} + k^{2\omega/p}/\eps^{(4/p-1)(2\omega+2)})$.
\end{theorem}
\begin{proof}
The correctness of the output is clear from the preceding lemmata by rescaling $\eps$. We discuss the runtime below.

We first examine the runtime to obtain $Z$. For $\eta_1$ in Lemma~\ref{lem:hat Q is good} we have $k\eta_1 \leq \eps$ and thus $K = \Theta(\eps/\eta_1)$. Applying Lemma~\ref{lem:add-mul_spec_approx}, we have $s = \tilde{O}(K/\eps^2)$ and obtaining the matrix $S$ takes time $O(\nnz(A)\log n) + \tilde{O}(mK^{\omega-1})$. By Lemma~\ref{lem:spec_approx}, we can obtain $(SA)T$ in time $O(\nnz(SA)) + \tilde{O}(s^\omega/\eps^2)$ for a matrix $T$ of $\tilde{O}(s/\eps^2)$ columns and thus the subsequent SVD of $SAT$, by Lemma~\ref{lem:thin_svd}, takes $\tilde{O}(s^\omega/\eps^2)$. These three steps take time $\tilde{O}(mK^{\omega-1}) + O(\nnz(SA)) + \tilde{O}(s^\omega/\eps^2) = O(\nnz(A)) + \tilde{O}(mK^2  + K^\omega/\eps^{2\omega+2})$, where we used the fact that $S$ samples the rows of $A$ without replacement and so $\nnz(SA)\leq \nnz(A)$. Calculating the row span of $W'SA$, which is a $k$-by-$n$ matrix, takes $O(nk^{\omega-1})$ time. The total runtime is $O(\nnz(A)\log n) + \tilde{O}_p(mK^{\omega-1} + K^\omega/\eps^{2\omega+2})$. Plugging in $K= \eps/\eta_1 = \Theta(k^{2/p}/\eps^{4/p-1})$ yields the runtime $O(\nnz(A)\log n) + \tilde{O}(mk^{2(\omega-1)/p}/\eps^{(4/p-1)(\omega-1)} + k^{2\omega/p}/\eps^{(4/p-1)(2\omega+2)})$.

Next we examine the runtime to obtain $\hat Y$. Since $R$ has $t = \Theta(k/\eta_2)$ rows and $AR$ can be computed in $O(\nnz(A))$ time, $Z^TR$ can be computed in $O(nk)$ time, the row space of $Z^TR$ (which is a $k\times t$ matrix) in $O(k^{\omega-1}t) = \tilde O(k^\omega/\eta_2)$ time, and the final matrix product $(AR)P_{\text{rowspace}(Z^TR)}$ in $O(mt^{\omega-1}) = \tilde{O}(mk^{\omega-1}/\eta_2)$ time. Overall, computing $\hat Y$ takes time $O(\nnz(A)) + \tilde{O}(mk^{\omega-1}/\eta_2) = O(\nnz(A)) + \tilde{O}(mk^{\omega+2/p-2}/\eps^2)$.

The overall runtime follows immediately.
\end{proof}

\section{Case $p>2$}
The algorithm remains the same in Algorithm~\ref{alg:outline}. In this section we shall prove the correctness and analyse the runtime for constant $p>2$. The outline of the proof is the same and we shall only highlight the differences. Again we let $r=k/\eps$.

In place of Lemma~\ref{lem:(p,r)-norm-preserved}, we now have:
\begin{lemma}\label{lem:(p,r)-norm-preserved p>2}
Suppose that $p>2$ and $S$ satisfies \eqref{eqn:PCP_eta1}. It then holds for all rank-$k$ orthogonal projection $Q$ that
\begin{align*}
&\quad\ (1-K_p\eps)\norm{A(I-Q)}_{(p,r)}^p - C_{p/2,\eps}r\eta_1^{p/2}\norm{A-A_k}_F^p \\
&\leq \norm{SA(I-Q)}_{(p,r)}^p \\
&\leq (1+K_p\eps)\norm{A(I-Q)}_{(p,r)}^p + C_{p/2,\eps}r\eta_1^{p/2}\norm{A-A_k}_F^p,
\end{align*}
where $K_p\geq 1$ is some constant that depends only on $p$.
\end{lemma}
\begin{proof}
We now have two cases based on~\eqref{eqn:key_inequality}.
\begin{itemize}
	\item When $\sigma_i(A(I-Q))^2 \geq (1/\eps)\eta_1\|A-A_k\|_F^2$, we have 
	\[
		(1-O_p(\eps))\sigma_i^p(A(I-Q)) \leq \sigma_i^p(SA(I-Q)) \\
		\leq (1+O_p(\eps))\sigma_i^p(A(I-Q))
	\]
	\item When $\sigma_i(A(I-Q))^2 < (1/\eps)\eta_1\|A-A_k\|_F^2$, we have from Lemma~\ref{lem:elementary} that
	\begin{align*}
	&\quad\ (1-\eps)\sigma_i^p(A(I-Q)) - C_{p/2,\eps}\eta_1^{p/2} \norm{A-A_k}_F^p \\
	&\leq \sigma_i^p(SA(I-Q)) \\
	&\leq (1+\eps)\sigma_i^p(A(I-Q)) + C_{p/2,\eps}\eta_1^{p/2} \norm{A-A_k}_F^p.
	\end{align*}
\end{itemize}
The claimed result follows in the same manner as in the proof of Lemma~\ref{lem:(p,r)-norm-preserved}.
\end{proof}

The analogy of Lemma~\ref{lem:hat Q is good} is the following, where we apply H\"older's inequality that $\|A-A_k\|_F\leq n^{1/2-1/p}\|A-A_k\|_p$.
\begin{lemma}\label{lem:hat Q is good p > 2}
Let $\eps \in (0,1/2]$. Suppose that $\hat Q'$ satisfies \eqref{eqn:hat Q'}. Then it holds for some constants $c_p, c_p' > 0$ which depend only on $p$ that
\[
\norm{A(I-\hat Q')}_p^p \leq (1+c_p\eps)\norm{A-A_k}_p^p,
\]
whenever $\eta_1 \leq c_p'\eps^{1+2/p}/(k^{2/p} n^{1-2/p})$.
\end{lemma}
\begin{proof}
Similarly we have
\[
\norm{A(I-\hat Q')}_p^p \leq (1+c_p\eps)\norm{A-A_k}_p^p 
 + r C_{p/2,\eps}\eta_1^{p/2} n^{\frac{p}{2}-1}\norm{A-A_k}_p^p.
\]
The conclusion follows when 
\[
C_{p/2,\eps}\eta_1^{p/2} n^{p/2-1}\leq \frac{\eps}{r} = \frac{\eps^2}{k},
\]
that is,
\[
\left(\frac{p}{2}\left(1+\frac{1}{\eps}\right)^{\frac{p}{2}-1}\right)\eta_1^{\frac{p}{2}}n^{\frac{p}{2}-1}\leq \frac{\eps^2}{k}.\qedhere
\]
\end{proof}

The analogy of Lemma~\ref{lem: Y hat} is the following.
\begin{lemma}\label{lem:Y hat p > 2}
When $p>2$ is a constant, the matrix $\hat Y$ defined in Line~\ref{alg:line Y hat} of Algorithm~\ref{alg:outline} satisfies with probability at least $0.9$ that
\[
\norm{A-\hat{Y} Z^T}_p \leq (1+c_p\eps)\min_{Y: \rk(Y)= k} \norm{A-YZ^T}_p,
\]
for some constant that depends only on $p$, whenever $\eta_2 \leq \eps^2/n^{1-2/p}$.
\end{lemma}
\begin{proof}
The proof is similar to that of Lemma~\ref{lem: Y hat} except that we have instead in the last part of the argument that
\begin{align*}
\norm{A - \hat Y Z}_p &\leq (1+c_p\eps)\norm{A - AZZ^T}_p + \norm{\hat Y-AZ}_p \\
										 & \leq (1+c_p\eps)\norm{A - AZZ^T}_p + \norm{\hat Y-AZ}_F \\
										 & \leq (1+c_p\eps)\norm{A - AZZ^T}_p + \sqrt{\eta_2}\norm{A-AZZ^T}_F \\
										 & \leq (1+c_p\eps)\norm{A - AZZ^T}_p + \sqrt{\eta_2} n^{\frac12-\frac1p}\norm{A-AZZ^T}_p 
\end{align*}
and we would need $\eta_2\leq \eps^2/n^{1-\frac2p}$. 
\end{proof}

In summary, we have the following main theorem.
\begin{theorem}\label{thm:p>2}
Let $p > 2$ be a constant. Suppose that $A\in\R^{m\times n}$ ($m\geq n$). There is a randomized algorithm which outputs $Y\in\R^{m\times k}$ and $Z\in \R^{n\times k}$ such that $\hat X = Y Z^T$ satisfies \eqref{eqn:low_rank_approx} with probability at least $0.97$. The algorithm runs in time $O(\nnz(A)\log n) +  \tilde{O}_p(n^{\omega(1-2/p)}k^{2\omega/p}/\eps^{2\omega/p+2} + mn^{(\omega-1)(1-2/p)}(k/\eps)^{2(\omega-1)/p})$. 
\end{theorem}
\begin{proof}
The correctness follows from the previous lemmata as in the proof of Theorem~\ref{thm:p<2}. Below we discuss the running time.

First we examine the time required to obtain $Z$. It is easy to verify that $k\eta_1\leq \eps$ and so $K = \Theta(\eps/\eta_1)$. Similar to the analysis in Theorem~\ref{thm:p<2}, we have the total runtime $O(\nnz(A)\log n) + \tilde{O}_p(mK^{\omega-1} + K^\omega/\eps^{2\omega+2}))$. Note that $K=\Theta(\eps/\eta_1) = \Theta(n^{1-2/p}(k/\eps)^{2/p})$, so the runtime becomes $O(\nnz(A)\log n) + \tilde{O}_p(mn^{(\omega-1)(1-2/p)}(k/\eps)^{2(\omega-1)/p} + n^{\omega(1-2/p)}k^{2\omega/p}/\eps^{2\omega/p+2})$.

Next we examine the time required to obtain $\hat Y$. Again similarly the runtime is $O(\nnz(A)) + \tilde{O}(mk^{\omega-1}/\eta_2) = O(\nnz(A)) + \tilde{O}(mn^{1-2/p}k^{\omega-1}/\eps^2)$.

Combining the two runtimes above yields the overall runtime.
\end{proof}

\section{Experiments}\label{sec:experiments}

The contribution of our work is primarily theoretical: an input sparsity time algorithm for low-rank approximation for any Schatten $p$-norm. In this section, nevertheless, we give an empirical verification of the advantage of our algorithm on both synthetic and real-world data. We focus on the most important case of the nuclear norm, i.e., $p=1$. 

In addition to the solution provided by our algorithm, we also consider a natural candidate for a low-rank approximation algorithm, which is the solution in Frobenius norm, that is, a rank-$k$ matrix $X$ for which $\|A-X\|_F \leq (1+\eps)\|A-A_k\|_F$. This problem admits a simple solution as follows. Take $S$ to be a \text{Count-Sketch} matrix and let $Z$ be an $n\times k$ matrix whose columns form an orthonormal basis of the top-$k$ right singular vectors of $SA$. Then $X=AZZ^T$ is a Frobenius-norm solution with high probability~\cite{CEMMP15}.

We shall compare the quality (i.e., approximation ratio measured in Schatten $1$-norm) of both solutions and the running times\footnote{All tests are run under MATLAB 2019b on a machine of Intel Core i7-6550U CPU@2.20GHz with 2 cores.}.

\paragraph{Synthetic Data.} 
We adopt a simpler version of Algorithm~\ref{alg:outline} by taking $S$ to be a \textsc{Count-Sketch} matrix of target dimension $k^2$ and both $R$ and $T$ to be identity matrices of appropriate dimension. For the Frobenius-norm solution, we also take $S$ to be a \textsc{Count-Sketch} matrix of target dimension $k^2$. We choose $n=3000$ and generate a random $n\times n$ matrix $A$ of independent entries, each of which is uniform in $[0,1]$ with probability $0.05$ and $0$ with probability $0.95$. Since the regime of interest is $k\ll n$, we vary $k$ among $\{5,10,20\}$. For each value of $k$, we run our algorithm $50$ times and record the relative approximation error $\eps_1 = \|A-YZ^T\|_1/\|A-A_k\|_1 - 1$ with the running time and the relative approximation error of the Frobenius-norm solution $\eps_2 = \|A-X\|_1/\|A-A_k\|_1 - 1$ with the running time. The same matrix $A$ is used for all tests. In Table~\ref{tab:comparison} we report the median of $\eps_1$, the median of $\eps_2$, the median running time of both algorithms, among $50$ independent runs for each $k$, and the median running time of a full SVD of $A$ among $10$ runs.
\begin{table}
\caption{Performance of our algorithm on synthetic data compared with the approximate Frobenius-norm solution and the SVD.}\label{tab:comparison}
\centering
\begin{tabular}{|c|c|c|c|}
\hline
& $k=5$ & $k=10$ & $k=20$ \\
\hline
median of $\eps_1$ & 0.00372 & 0.00377 & 0.00486 \\
\hline
median of $\eps_2$ & 0.00412 & 0.00485 & 0.00637 \\
\hline
\hline
median runtime of & \multirow{2}{*}{0.067s} & \multirow{2}{*}{0.196s} & \multirow{2}{*}{0.428s} \\
$\|\cdot\|_1$ algorithm & & & \\
\hline
median runtime of & \multirow{2}{*}{0.044s} & \multirow{2}{*}{0.073s} & \multirow{2}{*}{0.191s} \\
$\|\cdot\|_F$ algorithm & & & \\
\hline
median runtime of SVD & \multicolumn{3}{c|}{5.788s}\\
\hline
\end{tabular}
\end{table}

We can observe that our algorithm achieves a good (relative) approximation error, which is less than $0.005$ in all such cases of $k$. Our algorithm also outperforms the approximate Frobenius-norm solution by 10\%--30\% in terms of approximation error. Our algorithm also runs about 13-fold faster than a regular SVD.

\paragraph{KOS data.} For real-world data, we use a word frequency dataset, named KOS, from UC Irvine.\footnote{\url{https://archive.ics.uci.edu/ml/datasets/Bag+of+Words}} The matrix represents word frequencies in blogs and has dimension $3430\times 6906$ with $353160$ non-zero entries. Again we report the median relative approximation error and the median running time of our algorithm and those of the Frobenius-norm algorithm among $50$ independent runs for each value of $k\in \{5,10,20\}$. The results are shown in Table~\ref{tab:KOS}.

Our algorithm achieves a good approximation error, less than $0.015$, and surpasses the approximate Frobenius-norm solution for all such values of $k$. The gap between two solutions in the approximation error widens as $k$ increases. When $k=10$, our algorithm outperforms the approximate Frobenius-norm by $30\%$; when $k=20$, this increases to almost $50\%$. Our algorithm, although stably slower than the Frobenius norm algorithm by 30\%--40\%, still displays a 14.5-fold speed-up compared with the regular SVD.

\begin{table}
\caption{Performance of our algorithm on KOS data compared with the approximate Frobenius-norm solution and the SVD.}\label{tab:KOS}
\centering
\begin{tabular}{|c|c|c|c|}
\hline
& $k=5$ & $k=10$ & $k=20$ \\
\hline
median of $\eps_1$ & 0.0149 & 0.0145 & 0.0132 \\
\hline
median of $\eps_2$ & 0.0183 & 0.0216 & 0.0259 \\
\hline
\hline
median runtime of & \multirow{2}{*}{0.155s} & \multirow{2}{*}{0.204s} & \multirow{2}{*}{0.323s} \\
$\|\cdot\|_1$ algorithm & & & \\
\hline
median runtime of & \multirow{2}{*}{0.113s} & \multirow{2}{*}{0.154s} & \multirow{2}{*}{0.242s} \\
$\|\cdot\|_F$ algorithm & & & \\
\hline
median runtime of SVD & \multicolumn{3}{c|}{4.999s}\\
\hline
\end{tabular}
\end{table}

\section{Generalization} \label{sec:generalization}

More generally, one can ask to solve the problem of low-rank approximation with respect to some function $\Phi$ on the matrix singular values, i.e.,
\begin{equation}\label{eqn:general_low_rank_approx}
\min_{X: \rk(X) = k} \Phi(A-X)
\end{equation}
Here we consider $\Phi(A) = \sum_i \phi(\sigma_i(A))$ for some increasing function $\phi:[0,\infty)\to [0,\infty)$. It is clear that $\Phi$ is rotationally invariant and that $\Phi(A)\geq \Phi(B)$ if $\sigma_i(A)\geq \sigma_i(B)$ for all $i$. These two properties allow us to conclude that $A_k$ remains an optimal solution for such general $\Phi$.

We further assume that $\phi$ satisfies the following conditions.
\begin{enumerate}[label=(\alph*)]
	\item there exists $\alpha>0$ such that $\phi((1+\eps)x) \leq (1+\alpha\eps) \phi(x)$ and $\phi((1-\eps)x) \geq (1-\alpha\eps) \phi(x)$ 	for all sufficiently small $\eps$.
	\item it holds that for each sufficiently small $\eps$,
	\[
	K^1_{\phi,\eps} = \sup_{x > 0}\sup_{y \in [\eps x, x]} (\phi(x+y)-\phi(x))/\phi(y) < \infty
	\]
	and
	\[
	K^2_{\phi,\eps} = \sup_{x > 0}\sup_{y \in [\eps x, x]} (\phi(x)-\phi(x-y))/\phi(y) < \infty.
	\]
	\item it holds that for each sufficiently small $\eps$,
	\[
	L_{\phi,\eps} = \sup_{x > 0}  \phi(\eps x)/\phi(x)  < \infty.
	\]
	\item there exists $\gamma>0$ such that $\phi(x+y)\leq \gamma(\phi(x)+\phi(y))$.
\end{enumerate}
When the function $\phi$ is clear from the text, we also abbreviate $K^i_{\phi,\eps}$ and $L_{\phi,\eps}$ as $K^i_\eps$ and $L_\eps$, respectively. Let $K_\eps = \max\{K^1_\eps, K^2_\eps\}$.

It follows from a similar argument to Lemma~\ref{lem:(p,r)-norm-preserved p>2} and Conditions (a)--(c) that
\begin{align*}
&\quad\	(1-\alpha\eps) \phi(\sigma_i(A(I-Q))) - L_{\sqrt{\eta_1}} K_\eps \phi(\|A-A_k\|_F) \\
	&\leq \phi(\sigma_i(SA(I-Q))) \\
	&\leq (1+\alpha\eps) \phi(\sigma_i(A(I-Q))) + L_{\sqrt{\eta_1}} K_\eps \phi(\|A-A_k\|_F)
\end{align*}
Note Condition (c) implies that
$
\phi\big(\sqrt{\sum_i x_i^2}\big) \leq \phi\big(\sum_i x_i\big)\leq \gamma\sum_i \phi(x_i)
$,
which further implies that
$
\phi(\|A-A_k\|_F) \leq \gamma\Phi(A-A_k)
$.
Therefore
\begin{align*}
&\quad\	(1-\alpha\eps) \phi(\sigma_i(A(I-Q))) - \gamma L_{\sqrt{\eta_1}} K_\eps \Phi(A-A_k) \\
	&\leq \phi(\sigma_i(SA(I-Q))) \\
	&\leq (1+\alpha\eps) \phi(\sigma_i(A(I-Q))) + \gamma L_{\sqrt{\eta_1}} K_\eps \Phi(A-A_k)
\end{align*}

Analogously to the singular $(p,r)$-norm, we define $\Phi_r(A) = \sum_{i=1}^r \phi(\sigma_i(A))$. It is easy to verify that the argument of Lemmata~\ref{lem:find Q hat} to~\ref{lem:hat Q is good} will go through with minimal changes, yielding that
\[
\Phi_k(A(I-\hat Q')) \leq (1+c_1\eps)\Phi_k(A-A_k) + c_2 r\Phi(A-A_k)
\]
for some constants $c_1,c_2 > 0$ that depend on $\alpha,\gamma,K_\eps,L_\eps$. When $\eta_1 \leq c_3(\eps/r)^{1/\alpha}$ we have
\[
\Phi(A(I-\hat Q')) \leq (1+c_4\eps)\Phi_k(A-A_k).
\]
We can then output $AZ$ and $Z$ in time $O(\nnz(A)\cdot k + nk)$. Performing a similar analysis on the running time as before, we arrive at the following theorem.

\begin{theorem}
Suppose that $\phi:[0,\infty)\to[0,\infty)$ is increasing and satisfies Conditions (a)--(d) and that $K_\eps = \poly(1/\eps)$ and $L_\eps=\poly(1/\eps)$. Let $A\in\R^{n\times n}$. There is a randomized algorithm which outputs matrices $Y, Z\in \R^{n\times k}$ such that $X = YZ^T$ satisfies \eqref{eqn:general_low_rank_approx} with probability at least $0.98$. The algorithm runs in time $O(\nnz(A)(k+\log n)) + \tilde{O}(n\poly(k/\eps))$, where the hidden constants depend on $\alpha,\gamma$ and the polynomial exponents for $K_\eps$ and $L_\eps$.
\end{theorem}

We remark that a few common loss functions satisfy our conditions for $\phi$. These include the Tukey $p$-norm loss function $\phi(x) = x^p\cdot \mathbf{1}_{\{x\leq \tau\}} + \tau^p\cdot \mathbf{1}_{\{x>\tau\}}$, the $\ell_1$-$\ell_2$ loss function $\phi(x) = 2\sqrt{1+x^2/2}-1$ and the Huber loss function $\phi(x) = x^2/2\cdot\mathbf{1}_{\{x\leq \tau\}} + \tau(x-\tau/2)\cdot \mathbf{1}_{\{x > \tau\}}$.

\bibliographystyle{alpha}

\bibliography{literature}

\end{document}